\documentclass[aps,prxquantum,reprint,superscriptaddress,longbibliography,nofootinbib]{revtex4-2}

\usepackage{amsmath,amssymb,amsthm,bm,graphicx,booktabs}
\usepackage[colorlinks=true,linkcolor=blue,citecolor=blue,urlcolor=blue]{hyperref}

\newtheorem{theorem}{Theorem}
\newtheorem{lemma}[theorem]{Lemma}
\newtheorem{proposition}[theorem]{Proposition}
\newtheorem{corollary}[theorem]{Corollary}

\newcommand{\dd}{\,\mathrm{d}}
\newcommand{\CC}{\mathbb{C}}
\newcommand{\ZZ}{\mathbb{Z}}
\newcommand{\tr}{\operatorname{tr}}
\newcommand{\Om}{\Omega}
\newcommand{\Xt}{\widetilde X}

\begin{document}

\title{Time-optimal robust quantum control for partially known errors}

\author{Lloren\c{c} Balada Gaggioli}
\affiliation{Faculty of Electrical Engineering, Czech Technical University in Prague, Czech Republic}
\affiliation{LAAS-CNRS, Universit\'e de Toulouse, Toulouse, France}

\date{October 15, 2026}

\begin{abstract}
Robust quantum control usually assumes one of two extremes: the error is
of a known form and only its strength is unknown, or nothing at all is
known about it. In practice one often sits in between. The size of a
static error is unknown, but its \emph{class} is known, for example, it
only shifts energy levels, or it acts on each subsystem separately. We
ask how fast a control can cancel, to first order, every error in such a
class. A time-independent Hamiltonian can do this if and only if all
errors in the class have zero diagonal in a common basis. In that case an
equally spaced energy ladder in this basis does it after one period. A
lower bound that holds for any time-dependent control with bounded drive
power shows that this simple constant control is optimal for several
natural classes, including unknown level shifts and unknown local errors
on two subsystems. Completely unknown errors need a time-dependent
control, and there, knowing the error class shortens the minimum time by
a factor $\sqrt{d^2+1}$. So partial knowledge of the error changes not
only how long an optimal robust control takes, but also its form.
\end{abstract}

\maketitle

\section{Introduction}

Static or slowly drifting errors are among the most common imperfections of quantum hardware: a
transition frequency that has drifted since the last calibration, a stray magnetic field, a
residual coupling between neighbouring qubits.  A powerful way to suppress them is to shape the
control so that the error, seen from the frame that moves with the control, averages to zero.
This idea goes back to average Hamiltonian theory in nuclear magnetic
resonance~\cite{HaeberlenWaugh1968}.  It underlies dynamical
decoupling~\cite{ViolaLloyd1998,ViolaKnillLloyd1999}, its bounded-control
versions~\cite{ViolaKnill2003}, and dynamically corrected gates~\cite{KhodjastehViola2009}.  Because
any real drive has a bounded amplitude, the averaging takes time.  How much time depends on how much
is known about the error, and this dependence is the subject of this paper.

When the error operator is known, the problem is well understood.  Analytical conditions for smooth
pulses that cancel a known quasi-static error were derived in ~\cite{Barnes2015}, and all such
first-order corrections of a qubit can be described geometrically as closed
curves~\cite{Zeng2019,Buterakos2021}.  For a qubit with a specified error, the fastest robust pulses
have been found analytically~\cite{ZengBarnes2018}, by optimal-control
theory~\cite{VanDamme2017,FresseColson2025}, and numerically~\cite{Cao2024}.

That the unwanted interaction is often known only up to a class is an old observation.  Already the
original formulation of dynamical decoupling averages over a whole space of interaction
operators~\cite{ViolaKnillLloyd1999}, and decoupling can be selective, removing some interactions
while keeping others~\cite{ViolaLloydKnill1999}.  Efficient pulse sequences for structured couplings
follow from orthogonal arrays~\cite{StollsteimerMahler2001,RottelerWocjan2006}, symmetries of the
unwanted Hamiltonian~\cite{Read2025}, and classical codes, which give time-optimal pulsed sequences
for chosen sets of $k$-local interactions~\cite{Nguyen2026}.  In robust control, Poggi \emph{et al.}\
showed that first-order robustness against \emph{every} error is equivalent to the control path
forming a unitary $1$-design~\cite{Poggi2024}, and continuous paths with this property were
constructed in any dimension~\cite{Yang2026}.  Zeng and Deng bounded the number of distinct unitaries
that a schedule robust to a given noise space must visit~\cite{ZengDeng2025}.

What remains open is the exact answer for continuous, amplitude-limited control.  Given an error
class, what is the shortest time in which a bounded drive, with arbitrary time dependence, can
cancel every error in it to first order, and what does a control attaining that time look like?

The central message of this paper is that, for the error classes studied here, \emph{partial
knowledge of the error makes a time-independent Hamiltonian globally time optimal}.  The mechanism is
elementary [Fig.~\ref{fig:ladder}(a)].  Suppose every error in the class has zero diagonal in some
basis, and apply a constant Hamiltonian whose energy levels are equally spaced in that basis.  Every
error component then connects two different levels.  In the frame that moves with the control, it
oscillates at a nonzero multiple of the level spacing, and it averages to zero after one period.  Our
results make this precise:
\begin{enumerate}
\item A criterion: a constant control can cancel an error class if and only if the class has zero
diagonal in some basis.  The cheapest such control has equally spaced levels
(Sec.~\ref{sec:constant}).
\item A lower bound valid for \emph{every} control, time dependent or not.  The error class enters
only through the number of independent time modes it forces the propagator to occupy
(Sec.~\ref{sec:bound}).
\item For unknown level shifts of a $d$-level system, for unknown local errors on two
equal-dimensional subsystems, and for a qubit error in a plane, the two coincide: a constant
Hamiltonian attains the minimum duration $\Om T=2\pi\sqrt{d(d^2-1)/6}$.  For completely unknown errors the bound is larger and no constant control can reach it, knowing the error class therefore reduces the minimum time by exactly $\sqrt{d^2+1}$
(Table~\ref{tab:results}).
\end{enumerate}
The optimal controls are physically simple (Sec.~\ref{sec:controls}).  For a qudit with unknown
level shifts it is a constant chiral drive of all transitions.  For two qubits with unknown local
fields it is a constant anisotropic two-body coupling, which beats every combination of single-qubit
drives.

Cancelling static errors by toggling-frame averaging with continuous fields, decoupling and unitary-design robustness are well established, and
counting arguments already tell how many distinct unitaries a robust schedule needs~\cite{ZengDeng2025}.
What is new here is a continuous-time lower bound that is set by the structure of the error class and
holds for arbitrary time-dependent controls, together with simple controls that attain it exactly.
The minimum times are therefore global optima, not optima within a pulse ansatz.

These controls are also distinct from earlier uses of constant or continuous fields against noise.
Continuous dynamical decoupling~\cite{Chen2006,Fanchini2007,FanchiniNapolitano2007}, including its
qudit version~\cite{Napolitano2021}, uses time-dependent fields.  Constant dressing fields protect
selected transitions or encoded states~\cite{Timoney2011,Trypogeorgos2018,Barfuss2018,Kolbl2019},
exchange couplings protect some entangled states~\cite{De2011}, and energy penalties protect encoded
subspaces when the gap is large~\cite{BaconBrownWhaley2001,JordanFarhiShor2006,Bookatz2015,MarvianLidar2014};
in the strong-coupling limit these effects reduce to a projection onto the eigenspaces of the strong
Hamiltonian~\cite{Facchi2004,Burgarth2019}.  The controls derived here instead protect the full
identity operation, cancel every error in the class to first order at a finite time, and that time is
the shortest possible.

\begin{figure*}[t]
\includegraphics[width=\textwidth]{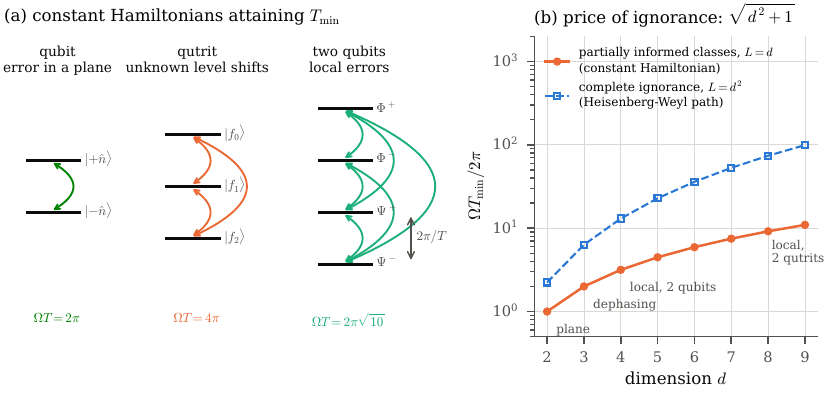}
\caption{Constant controls for partially known errors.  (a)~Horizontal lines are the energy levels of
the constant control Hamiltonian, spaced by $2\pi/T$.  Arrows are the matrix elements of the possible
errors in that basis.  Since every error has zero diagonal there, each component connects two
different levels; in the frame moving with the control it oscillates at a nonzero multiple of
$2\pi/T$ and integrates to zero over one period (the arrows are oscillating matrix elements, not
transitions).  Left: a qubit error confined to a plane, cancelled by a $2\pi$ rotation about the
normal $\hat n$.  Middle: unknown level shifts of a qutrit, cancelled by a Hamiltonian diagonal in the
Fourier basis $|f_k\rangle$, a chiral drive of all transitions.  Right: unknown local errors on two
qubits, cancelled by an anisotropic $XX$--$ZZ$ coupling diagonal in the Bell basis.  (b)~Minimum
duration~\eqref{eq:Tmin} for the partially informed classes (plane, level shifts, local errors on two
equal subsystems), compared with that for a completely unknown error~\eqref{eq:Tuniv}, which requires
time-dependent control and is $\sqrt{d^2+1}$ times longer.}
\label{fig:ladder}
\end{figure*}

\section{Setting}

\paragraph*{Control.}  We consider a $d$-level system whose Hamiltonian $H(t)$ we choose, with
arbitrary (bounded, measurable) time dependence.  Its \emph{propagator} $U(t)$ solves $\dot U=-iHU$
with $U(0)=I$.  We expand $H(t)=\frac12\sum_au_a(t)\lambda_a$ in the standard traceless Hermitian
basis matrices $\lambda_a$ (Pauli matrices for $d=2$, Gell-Mann matrices for $d=3$), normalized as
$\tr(\lambda_a\lambda_b)=2\delta_{ab}$, and bound the control amplitude,
\begin{equation}
 |\bm u(t)|^2=2\|H(t)\|_F^2\le\Om^2 ,
 \label{eq:bound}
\end{equation}
where $\|A\|_F^2=\tr(A^\dagger A)$.  For a qubit, $\bm u$ is the Rabi vector and $\Om$ the maximal Rabi
frequency.  In general, Eq.~\eqref{eq:bound} bounds the total drive power summed over all couplings.
It is not a bound on each coupling separately, nor on the operator norm; the criteria below do not
depend on this choice, but the numerical values of the minimum times do (Sec.~\ref{sec:scope}).  We
look for controls that implement the identity at a final time $T$, $U(T)=e^{i\theta}I$ up to an
irrelevant global phase.  Such controls serve as robust idle operations and as decoupling blocks.

\paragraph*{Errors.}  The actual Hamiltonian is $H(t)+\epsilon X$, where $\epsilon X$ is a small
static error.  The operator $X$ is unknown, but it is known to belong to an \emph{error class} $S$, a
real linear space of traceless Hermitian operators.  Three classes will be central:
\begin{itemize}
\item \emph{Error in a plane} (qubit): $X=\bm v\cdot\bm\sigma$ with $\bm v$ perpendicular to a known
axis $\hat n$, for example a transverse stray field of unknown direction.
\item \emph{Level shifts} (qudit): $X$ is diagonal in a known basis $\{|j\rangle\}$.  The error shifts
the energy levels by unknown amounts but does not drive transitions.  This describes frequency drifts
and the inhomogeneous spread of level shifts across an ensemble (static dephasing), not Markovian
decoherence.
\item \emph{Local errors} on two subsystems of equal dimension $d_A$, so $d=d_A^2$: $X=A\otimes
I+I\otimes B$ with arbitrary traceless $A$ and $B$.  Each subsystem feels its own unknown field, and
there is no unknown coupling between them.
\end{itemize}
The reference limit is \emph{complete ignorance}, where $S$ contains every traceless operator.

\paragraph*{Robustness.}  The error acts on the system while the control rotates it, so its effect is
best seen in the \emph{toggling frame}, the frame that moves with the control.  There the error reads
$U^\dagger(t)XU(t)$, and to first order in $\epsilon$ the final propagator is
\begin{equation}
 U_\epsilon(T)=U(T)\Big[I-i\epsilon\int_0^TU^\dagger(t)XU(t)\dd t+O(\epsilon^2)\Big].
 \label{eq:firstorder}
\end{equation}
We call a control \emph{robust to $S$} if the first-order term vanishes for every error in the
class,
\begin{equation}
 \int_0^TU(t)^\dagger XU(t)\dd t=0\qquad\text{for every }X\in S .
 \label{eq:robust}
\end{equation}
This is a first-order condition: it cancels the leading contribution of the error, not the error
itself.  The second-order term generally survives, so the infidelity
$1-|\tr(U^\dagger(T)U_\epsilon(T))/d|^2$ of a robust control scales as $\epsilon^4$ instead of
$\epsilon^2$ (Appendix~\ref{app:magnus}).  Because Eq.~\eqref{eq:robust} is linear in $X$, it also
cancels, to first order, a coupling $\sum_aX_a\otimes B_a$ ($X_a\in S$) to a slowly varying
environment.

\section{Constant controls}\label{sec:constant}

Consider a time-independent Hamiltonian $H=\sum_kE_kP_k$ with distinct energies $E_k$ and eigenspace
projectors $P_k$.  Its propagator is $U(t)=\sum_ke^{-iE_kt}P_k$.  In the toggling frame, each block of
the error rotates at the corresponding energy difference, so
\begin{equation}
 \int_0^TU^\dagger XU\dd t=\sum_{k,l}\Big(\int_0^Te^{i(E_k-E_l)t}\dd t\Big)P_kXP_l .
 \label{eq:constH}
\end{equation}
The time integral vanishes exactly when $(E_k-E_l)T$ is a nonzero multiple of $2\pi$.  Blocks within
one energy level, including all diagonal elements, do not rotate and never average out.  A constant
Hamiltonian is therefore robust to $S$ if and only if
\begin{itemize}
\item[(a)] every $X\in S$ vanishes inside each energy level, $P_kXP_k=0$ (for a degenerate level this
means the whole block, not only its diagonal), and
\item[(b)] every pair of levels that $S$ connects is split by a nonzero multiple of $2\pi/T$.
\end{itemize}
Returning to the identity, $U(T)\propto I$, requires in addition that \emph{all} pairs of levels be
split by multiples of $2\pi/T$.

This costs drive power: the levels must be spread apart, which requires a large $\|H\|_F$.  The
cheapest arrangement follows from an elementary counting lemma, which we will use twice.

\begin{lemma}[Equal spacing is cheapest]\label{lem:spacing}
Let $y_1<\dots<y_N$ be distinct real numbers whose differences are integers, and let $\bar y$ be
their mean.  Then $\sum_i(y_i-\bar y)^2\ge N(N^2-1)/12$, with equality if and only if consecutive
values differ by exactly one.  Consequently $\sum_iy_i^2\ge N(N^2-1)/12$, with equality if and only
if the $y_i$ are $-\frac{N-1}2,-\frac{N-3}2,\dots,\frac{N-1}2$.
\end{lemma}

\begin{proof}
Distinct numbers with integer differences satisfy $|y_i-y_j|\ge|i-j|$, since $y_i-y_j$ is a sum of
$|i-j|$ positive integer gaps.  Therefore
$\sum_i(y_i-\bar y)^2=\frac1{2N}\sum_{i,j}(y_i-y_j)^2\ge\frac1{2N}\sum_{i,j}(i-j)^2=\frac{N(N^2-1)}{12}$,
with equality iff all gaps equal one.  The second statement follows from
$\sum_iy_i^2=\sum_i(y_i-\bar y)^2+N\bar y^2$.
\end{proof}

\begin{theorem}[Constant robust control]\label{thm:constant}
(i) A constant Hamiltonian robust to $S$ that returns to the identity exists if and only if there is
an orthonormal basis in which every $X\in S$ has zero diagonal.  (ii) Suppose such a Hamiltonian has
$K$ distinct energy levels, each with the same degeneracy $d/K$.  Then
\begin{equation}
 \Om T\ \ge\ 2\pi\sqrt{\frac{d\,(K^2-1)}{6}},
 \label{eq:law}
\end{equation}
with equality if and only if the levels are equally spaced by $2\pi/T$ and $|\bm u|=\Om$.
\end{theorem}

\begin{proof}
(i) If $H$ is robust, condition (a) holds, so every $X\in S$ has zero diagonal in any eigenbasis of
$H$.  Conversely, if every $X\in S$ has zero diagonal in a basis $\{|e_k\rangle\}$, then the
nondegenerate ladder
\begin{equation}
 H=\frac{2\pi}T\sum_{k=0}^{d-1}\Big(k-\frac{d-1}2\Big)|e_k\rangle\langle e_k|
 \label{eq:ladder}
\end{equation}
satisfies (a), (b) and $U(T)\propto I$.  (ii) Write the energies as $E=2\pi y/T$.  Returning to the
identity requires the $K$ values $y_k$ to differ by nonzero integers.  Since $H$ is traceless and the
degeneracies are equal, $\bar y=0$, and
$\|H\|_F^2=(2\pi/T)^2\frac dK\sum_ky_k^2\ge(2\pi/T)^2\frac{d(K^2-1)}{12}$ by Lemma~\ref{lem:spacing}.
With $|\bm u|^2=2\|H\|_F^2\le\Om^2$ this gives Eq.~\eqref{eq:law}.
\end{proof}

Part (ii) assumes equal degeneracies; we do not determine the cheapest constant spectrum for
arbitrary error classes.  For the classes of interest, the ladder~\eqref{eq:ladder} is nondegenerate,
$K=d$.

\paragraph*{When does a zero-diagonal basis exist?}  Bringing a set of matrices to zero diagonal by
one common change of basis is known in matrix theory as \emph{simultaneous hollowization}.  A single
traceless matrix can always be brought to zero diagonal~\cite{Fillmore1969}, so every \emph{single}
known error, in any dimension, can be cancelled by a constant control.  For two or more matrices the
problem is subtle~\cite{DammFassbender2020,Nicholus2026}.  Theorem~\ref{thm:constant}(i) shows that
this mathematical problem decides whether time dependence is necessary in robust control.  One
obstruction is easy to state: if $S$ contains a nonzero operator that is diagonal in the candidate
basis, that basis fails.  For complete ignorance, $S$ contains $|e_1\rangle\langle e_1|-I/d$ for
every basis, so complete ignorance \emph{forces} time-dependent control.  For a qubit, a zero-diagonal
basis exists exactly when $S$ is at most two-dimensional.

\paragraph*{The three classes.}  All three classes of interest admit a zero-diagonal basis.
\begin{itemize}
\item \emph{Plane.}  In the eigenbasis of $\hat n\cdot\bm\sigma$, every Pauli operator perpendicular
to $\hat n$ is off-diagonal.  The constant control is a $2\pi$ rotation about $\hat n$.
\item \emph{Level shifts.}  In the Fourier basis $|f_k\rangle=d^{-1/2}\sum_je^{2\pi ijk/d}|j\rangle$,
every diagonal operator has diagonal elements $\langle f_k|X|f_k\rangle=\tr X/d=0$.  The constant
control is
\begin{equation}
 H_{\rm D}=-\frac{2\pi}T\sum_{k=0}^{d-1}\Big(k-\frac{d-1}2\Big)|f_k\rangle\langle f_k| .
 \label{eq:HD}
\end{equation}
\item \emph{Local errors.}  In any basis of maximally entangled states, every local operator has
diagonal elements $\langle\Phi|A\otimes I|\Phi\rangle=\tr A/d_A=0$, because each half of a maximally
entangled state is maximally mixed.  For two qubits this is the Bell basis $\{\Phi^\pm,\Psi^\pm\}$; for
general $d_A$ one can use the states $(I\otimes \mathsf X^p\mathsf Z^q)\sum_j|jj\rangle/\sqrt{d_A}$,
with $\mathsf X$ the cyclic shift and $\mathsf Z$ the clock operator.  The constant control is a
ladder~\eqref{eq:ladder} in this basis.
\end{itemize}

\section{No time-dependent control is faster}\label{sec:bound}

Theorem~\ref{thm:constant} compares constant controls only with each other.  A time-dependent pulse
might do better.  To rule this out we need a bound that holds for every control.  The idea is to treat
the $d^2$ matrix elements $U_{jk}(t)$ of the propagator as signals in time:
\begin{itemize}
\item robustness forces these signals to spread over many independent directions;
\item a bounded drive limits how fast they can oscillate;
\item signals that return to their initial value are built from Fourier harmonics of the period, and
high harmonics cost drive power.
\end{itemize}
Counting the harmonics then bounds the time.

\paragraph*{Signals and their overlaps.}  Define the normalized signals
$\phi_{jk}(t)=\sqrt{d/T}\,U_{jk}(t)$ and their overlap (Gram) matrix
\begin{equation}
 \Gamma_{(jm),(ln)}=\int_0^T\overline{\phi_{jm}(t)}\,\phi_{ln}(t)\dd t .
 \label{eq:Gamma}
\end{equation}
$\Gamma$ is a positive $d^2\times d^2$ matrix with trace $d^2$.  Three facts translate the control
problem into a statement about these signals.

\emph{Robustness.}  Since $[U^\dagger XU]_{mn}=\sum_{j,l}\overline{U_{jm}}X_{jl}U_{ln}$, the robustness
condition~\eqref{eq:robust} becomes a set of linear equations for $\Gamma$,
\begin{equation}
 \textstyle\sum_{j,l}X_{jl}\,\Gamma_{(jm),(ln)}=0\quad\text{for all }X\in S\text{ and all }m,n .
 \label{eq:relations}
\end{equation}
Unitarity of $U$ adds $\sum_j\Gamma_{(jm),(jn)}=d\,\delta_{mn}$.

\emph{Drive power.}  From $\dot U=-iHU$ and the unitary invariance of the Frobenius norm,
$\sum_{jk}|\dot U_{jk}|^2=\|H\|_F^2$.  The drive bound therefore limits the total ``kinetic energy'' of
the signals,
\begin{equation}
 \sum_{j,k}\int_0^T|\dot\phi_{jk}|^2\dd t=\frac d{2T}\int_0^T|\bm u|^2\dd t\le\frac{d\,\Om^2}2 .
 \label{eq:energy}
\end{equation}

\emph{Closure.}  If $U(T)=e^{i\theta}I$, every signal returns to its initial value up to the phase
$e^{i\theta}$.  Such signals are combinations of the harmonics $e^{2\pi iyt/T}$ with
$y\in\ZZ+\theta/2\pi$, and the harmonic $y$ has kinetic energy $(2\pi y/T)^2$.

\paragraph*{Occupation.}  The eigenvalues of $\Gamma$ measure how strongly the signals share
directions in the space of functions of time.  If $\Gamma\le g\,I$ for every control robust to $S$,
no single function, and in particular no single harmonic, can carry more than $g$ of the total weight
$d^2$.  We call $g$ the \emph{maximal occupation} allowed by the error class and write $L=d^2/g$.  The
signals must then spread over at least $L$ harmonics.

The picture is simple.  A larger error class imposes more equations~\eqref{eq:relations}, which leave
less room for the signals to concentrate.  The occupation $g$ becomes smaller, the propagator must
use more harmonics, and high harmonics cost time.  Note that $L$ counts \emph{time modes}, not energy
levels.  For a constant Hamiltonian with $K$ equally degenerate levels the two numbers agree, which
is why they coincide for the optimal controls below.  Enlarging a class can only lower $g$, but not
necessarily strictly; we give an example in Sec.~\ref{sec:controls}.

\begin{theorem}[Lower bound for all controls]\label{thm:lower}
If every control robust to $S$ has $\Gamma\le g\,I$, with $L=d^2/g$ an integer, then every robust
control that returns to the identity up to a phase, time dependent or not, satisfies
\begin{equation}
 \Om T\ \ge\ 2\pi\sqrt{\frac{d\,(L^2-1)}6} .
 \label{eq:lawL}
\end{equation}
It even satisfies the stronger energy bound $T\int_0^T|\bm u|^2\dd t\ge(2\pi)^2d(L^2-1)/6$.
\end{theorem}

\paragraph*{Idea of the proof.}  Decompose the signals into orthogonal time modes with weights
$\gamma_i\le g$ summing to $d^2$; these weights are the eigenvalues of $\Gamma$.  The total kinetic
energy~\eqref{eq:energy} is the weighted kinetic energy of the modes.  It is smallest when the weight
fills the lowest harmonics first, $g$ in each of the $L$ lowest ones.  Since the harmonics $y$ are
distinct with integer differences, Lemma~\ref{lem:spacing} bounds their energy by
$g(2\pi/T)^2L(L^2-1)/12$, and Eq.~\eqref{eq:energy} with $gL=d^2$ gives the result.  Making the
``fill the lowest harmonics'' step rigorous requires Ky Fan's variational
principle~\cite{KyFan1949,Bhatia} and an Abel summation; the full proof is in
Appendix~\ref{app:lower}.

Knowledge of the error enters only through $g$.  The next proposition computes it for the classes of
interest.

\begin{proposition}[What the error class implies]\label{prop:g}
(i) If $S$ contains every traceless operator diagonal in some basis, then $\Gamma\le d\,I$, so $L=d$.
(ii) If the system is a product $\CC^{d_A}\otimes\CC^{d_B}$ and $S$ contains every traceless operator
on the first factor, then $\Gamma\le d_B^2I$, so $L=d_A^2$.  (iii) If $S$ contains every traceless
operator, then $\Gamma=I$, so $L=d^2$.
\end{proposition}

The proof inserts suitable errors into Eq.~\eqref{eq:relations} and uses unitarity
(Appendix~\ref{app:g}).  All three bounds are attained by the controls below, so they cannot be
improved.

\begin{corollary}[A constant control is optimal]\label{cor:optimal}
Suppose $S$ has zero diagonal in some basis and contains either every traceless operator diagonal in
some basis, or every local operator $A\otimes I$ on two subsystems of equal dimension.  Then the
minimum duration of a robust identity is
\begin{equation}
 \Om T_{\min}=2\pi\sqrt{\frac{d\,(d^2-1)}6},
 \label{eq:Tmin}
\end{equation}
attained by the constant ladder~\eqref{eq:ladder} in the zero-diagonal basis.  This covers unknown
level shifts, unknown local errors on two equal-dimensional subsystems, and a qubit error in a plane.
No time-dependent control is faster.
\end{corollary}

\begin{proof}
Proposition~\ref{prop:g} gives $L=d$, so Theorem~\ref{thm:lower} gives Eq.~\eqref{eq:Tmin} as a lower
bound.  The ladder~\eqref{eq:ladder} is robust by Theorem~\ref{thm:constant} and has
$|\bm u|T=2\pi\sqrt{d(d^2-1)/6}$, so it attains the bound.  A plane of qubit errors contains the
level-shift class of the eigenbasis of any in-plane Pauli operator, so (i) applies with $d=2$.
\end{proof}

The corollary shows that a constant Hamiltonian \emph{attains} the global minimum time; it does not
exclude time-dependent controls of the same duration.

\paragraph*{The price of ignorance.}  For complete ignorance, Proposition~\ref{prop:g}(iii) gives
$L=d^2$ and the bound
\begin{equation}
 \Om T\ge2\pi\sqrt{\frac{d\,(d^4-1)}6}.
 \label{eq:Tuniv}
\end{equation}
No constant control can reach it, since no zero-diagonal basis exists.  It is nevertheless attained in
every dimension by a time-dependent path.  Let $D=\sum_j(j-\frac{d-1}2)|j\rangle\langle j|$, let $\mathcal F$
be the discrete Fourier transform, $B=\mathcal FD\mathcal F^\dagger$ and $\nu=2\pi/T$.

\begin{proposition}[Complete ignorance]\label{prop:universal}
The path $U_*(t)=e^{-i\nu tB}e^{-id\nu tD}$, generated by the smooth Hamiltonian
$H_*(t)=\nu[B+d\,e^{-i\nu tB}De^{i\nu tB}]$, cancels every traceless error to first order, ends at
$U_*(T)=(-1)^{d-1}I$, and has constant amplitude with $|\bm u_*|T=2\pi\sqrt{d(d^4-1)/6}$.  It therefore
attains the bound~\eqref{eq:Tuniv}.
\end{proposition}

The key point is that the $d^2$ matrix elements of $U_*$ oscillate at $d^2$ distinct consecutive
harmonics, each with equal weight, which is exactly the optimal filling of Theorem~\ref{thm:lower}
(Appendix~\ref{app:universal}).  This path is, up to conventions and the removal of a trace term, the
Heisenberg--Weyl continuous design of ~\cite{Yang2026}; the new statement is that it is globally
time optimal under the bound~\eqref{eq:bound}.

Since both minimum times are exact, the ratio
\begin{equation}
 \frac{T^{\min}_{\rm unknown}}{T^{\min}_{\rm class}}=\sqrt{\frac{d^4-1}{d^2-1}}=\sqrt{d^2+1}
 \label{eq:price}
\end{equation}
is exactly the time saved by knowing the error class [Fig.~\ref{fig:ladder}(b)].

\paragraph*{Relation to earlier bounds.}  The integer $L$ coincides with known lower bounds on the
number of distinct unitaries in a robust schedule~\cite{ZengDeng2025}: $d$ for a level-shift class,
and $q^2$ when the class contains all operators on a $q$-dimensional subsystem.  For complete
ignorance, $d^2$ is the minimal size of a unitary $1$-design~\cite{GrossAudenaertEisert2007}.
Theorem~\ref{thm:lower} turns these counts into exact times.  Under an operator-norm bound
$\|H\|\le u_{\max}$ it implies $u_{\max}T\ge\pi\sqrt{(L^2-1)/3}$, which for a single known error
($L=2$) reproduces the bound $\pi/u_{\max}$ of ~\cite{ZengDeng2025}.

The Fourier-filling argument itself was used by King and Mixon to bound the length of spherical design
curves~\cite{KingMixon2026,EhlerGrochenig2023}.  The new ingredient here is the occupation bound, which
adapts the argument to error classes.

\begin{table}[b]
\caption{Fastest robust identity under $|\bm u|\le\Om$.  All times are exact minima.  The first four
rows are attained by constant Hamiltonians; the last row requires time-dependent control.}
\label{tab:results}
\begin{ruledtabular}
\begin{tabular}{lcl}
error class & $L$ & $\Om T_{\min}$\\
\hline
qubit, error in a plane & $2$ & $2\pi$\\
qudit, unknown level shifts & $d$ & $2\pi\sqrt{d(d^2-1)/6}$\\
two $d_A$-level systems, local & $d_A^2$ & $2\pi\sqrt{d_A^2(d_A^4-1)/6}$\\
two qubits, local $+\,ZZ$ & $4$ & $2\pi\sqrt{10}$\\
\hline
completely unknown & $d^2$ & $2\pi\sqrt{d(d^4-1)/6}$
\end{tabular}
\end{ruledtabular}
\end{table}

\section{The optimal controls}\label{sec:controls}

\paragraph*{Qubit.}  An error known to lie in a plane, such as a transverse field of unknown direction,
is cancelled by a single $2\pi$ rotation about the normal of the plane, in $\Om T=2\pi$.  This is the
same time as for a single known error.  A third unknown direction removes the zero-diagonal basis,
forces time-dependent control, and raises the minimum time by a factor $\sqrt5$.

\paragraph*{Qudit with unknown level shifts.}  Unknown level shifts are the dominant static error of
many multilevel systems, from Zeeman and Stark shifts in trapped ions to frequency drifts in
superconducting circuits.  In the computational basis, the optimal control~\eqref{eq:HD} has no
diagonal elements.  For a qutrit,
using $\langle j|f_k\rangle=e^{2\pi ijk/3}/\sqrt3$,
\begin{equation}
 H_{\rm D}=\frac{2\pi}{\sqrt3\,T}\begin{pmatrix}0&e^{-i\pi/6}&e^{i\pi/6}\\
 e^{i\pi/6}&0&e^{-i\pi/6}\\ e^{-i\pi/6}&e^{i\pi/6}&0\end{pmatrix}.
 \label{eq:HDqutrit}
\end{equation}
The individual phases depend on how the basis states are chosen, but the phase around the loop does
not: $\arg(H_{01}H_{12}H_{20})=-\pi/2$, and $+\pi/2$ for the opposite orientation.  This loop phase
singles out a direction of circulation, which is the sense in which the drive is \emph{chiral}.  At
$t=T/3$ and $2T/3$ the evolution equals a cyclic shift of the levels, $|j\rangle\to|j-1\rangle$, up to a
phase, and at $t=T$ the qutrit is back at the identity; the motion is coherent and reversed by
$H\to-H$.  The chiral drive cancels any level shifts in $\Om T=4\pi$.  It is the continuous,
constant-amplitude counterpart of the shift-based decoupling sequences used against qudit
dephasing~\cite{Tripathi2025}.  It requires a simultaneous, phase-controlled drive of all three
transitions, as demonstrated with an NV spin~\cite{Barfuss2018,Kolbl2019}.
Cancelling arbitrary qutrit errors instead would require a time-dependent control $\sqrt{10}$ times
longer.

\paragraph*{Two qubits with unknown local fields.}  Local fields of unknown size and orientation on two
qubits are cancelled optimally by the constant coupling
\begin{equation}
 H_{\rm ex}=J\big(Z\otimes Z+\tfrac12X\otimes X\big),\qquad T=2\pi/J .
 \label{eq:Hex}
\end{equation}
This is an anisotropic two-body coupling, neither the isotropic Heisenberg exchange $XX+YY+ZZ$ nor the
$XY$ exchange $XX+YY$; we call it the $XX$--$ZZ$ coupling.  Its eigenstates are the four Bell states,
with equally spaced energies $J\{\frac32,\frac12,-\frac12,-\frac32\}$ for $\Phi^+,\Phi^-,\Psi^+,\Psi^-$.
Every local field connects only different Bell states, so it averages to zero after one period, at
$\Om T=2\pi\sqrt{10}$, and no single-qubit drive is needed.  Up to local rotations and relabelling of
axes, this is the only coupling $aXX+bYY+cZZ$ with equally spaced Bell energies
(Appendix~\ref{app:twoqubit}).  The isotropic Heisenberg coupling does not work, because its triplet is
degenerate and local errors act inside it.

The coupling protects more than the local errors: it also cancels the two-body errors
$XY,XZ,YX,YZ,ZX,ZY$, which have zero diagonal in the Bell basis.  It does not cancel $XX$, $YY$ or $ZZ$,
which are diagonal there; in particular, unknown $ZZ$ crosstalk is not protected.  This does not mean
that $ZZ$ crosstalk costs more time.  A Hadamard on the second qubit rotates the coupling into
\begin{equation}
 H'=J\big(Z\otimes X+\tfrac12X\otimes Z\big),
 \label{eq:Hrot}
\end{equation}
whose eigenstates are still maximally entangled and in which $ZZ$ has zero diagonal.  $H'$ cancels all
local errors and $ZZ$ in the same optimal time $2\pi\sqrt{10}/\Om$, since adding $ZZ$ keeps the local
errors in the class and hence $L=4$.  Enlarging the error class has not increased the minimum time.
Time dependence becomes necessary only when no common zero-diagonal basis remains, for instance for
fully arbitrary two-qubit errors ($L=16$).

Single-qubit drives cannot match the coupling.  With $H=H_1\otimes I+I\otimes H_2$, each qubit must
cancel every error acting on it, which is complete ignorance for that qubit.  Theorem~\ref{thm:lower}
then gives $\Om T\ge2\pi\sqrt{20}$, and this time is attained by running the $d=2$ path of
Proposition~\ref{prop:universal} on each qubit (Appendix~\ref{app:twoqubit}).

\begin{proposition}[Entangling advantage]\label{prop:local}
At equal total drive power~\eqref{eq:bound}, the $XX$--$ZZ$ coupling cancels unknown local fields
$\sqrt2$ times faster than the fastest control built from single-qubit drives.
\end{proposition}

The factor $\sqrt2$ refers to the bound~\eqref{eq:bound}, which assigns the same cost to a two-body
coupling as to a local field of equal Frobenius norm.  It is not universal under other hardware
constraints.  The coupling~\eqref{eq:Hex} is assumed to be available directly; if a platform must
synthesize it from other interactions, that overhead is not included.
In contrast to energy-penalty schemes, no information is encoded and no large gap is needed.  In
contrast to protection of particular entangled states by exchange~\cite{De2011}, every two-qubit state
is protected.

\begin{figure*}[t]
\includegraphics[width=\textwidth]{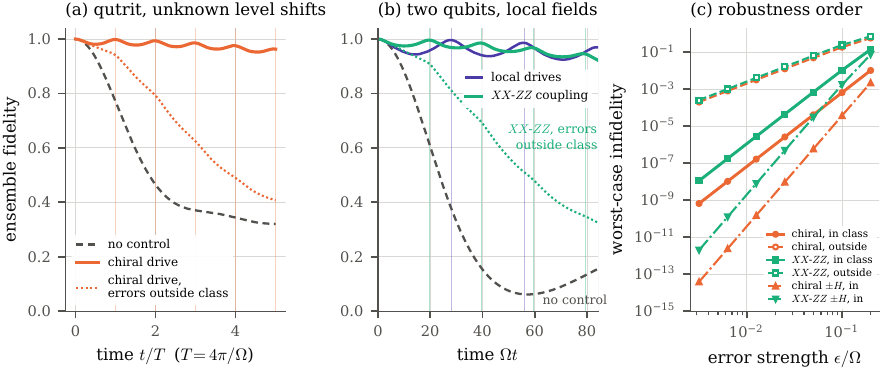}
\caption{Robust control as refocusing.  (a,b)~Ensembles of systems, each with its own random static
error from the stated class (root-mean-square strength $\epsilon=0.1\,\Om$).  The curves show the
ensemble-averaged fidelity $|\tr(U_0^\dagger U_\epsilon)/d|^2$ between the actual and ideal evolution;
vertical lines mark the control periods.  (a)~Qutrits with unknown level shifts: without control
(dashed) the ensemble dephases, while the constant chiral drive~\eqref{eq:HDqutrit} refocuses it at
every period $T=4\pi/\Om$.  Errors outside the class (dotted) are not refocused.  (b)~Two qubits with
unknown local fields: the $XX$--$ZZ$ coupling~\eqref{eq:Hex} refocuses at multiples of
$2\pi\sqrt{10}/\Om$, a factor $\sqrt2$ earlier than the fastest local drives ($2\pi\sqrt{20}/\Om$).  The
slow degradation of the peaks is the residual second-order error.  (c)~Worst-case infidelity at one
period versus error strength: it scales as $\epsilon^4$ for errors inside the class, as $\epsilon^2$
outside it, and as $\epsilon^6$ for the $\pm H$ sequence of Sec.~\ref{sec:scope}.  Simulation details
are given in Appendix~\ref{app:numerics}.}
\label{fig:revivals}
\end{figure*}

\paragraph*{Numerical checks.}  We verified every construction by exact numerical propagation.  Each
satisfies the robustness condition~\eqref{eq:robust} to $10^{-14}$ and reaches the predicted $|\bm u|T$
to machine precision, for level shifts in $d=2$ to $6$, for local errors with $d_A=2,3$, for local
errors plus $ZZ$, and for the complete-ignorance path in $d=2$ to $5$.  Fitting the infidelities of
Fig.~\ref{fig:revivals}(c) gives exponents $4.00$ inside the class, $2.00$ outside, and $5.99$ for the
$\pm H$ sequence (Table~\ref{tab:slopes}).  As an independent test of the lower bound, we minimized $T$
over general piecewise-constant controls from random initial pulses, for qutrit level shifts and
two-qubit local errors.  No robust control was shorter than Eq.~\eqref{eq:Tmin}; the optimizer settled
in local minima $1.3$ to $1.7$ times longer, so these runs probe the bound but do not test its
tightness, which the explicit constructions establish.

\section{Scope and open questions}\label{sec:scope}

\paragraph*{Drive normalization.}  The bound~\eqref{eq:bound} limits the total drive power over all
couplings.  It is natural for multi-tone drives, but it is a modelling choice.  Other bounds, for
example on each coupling separately or on the operator norm, change the numerical constants, including
the $\sqrt{d^2+1}$ and $\sqrt2$ factors.  They do not change the zero-diagonal criterion, the averaging
mechanism, or the fact that the controls above are robust.

\paragraph*{Higher orders.}  Our controls cancel static errors to first order.  The surviving
second-order term accumulates coherently from one period to the next; across an ensemble this shows up
as the slow decay of the peaks in Fig.~\ref{fig:revivals}.  It can be removed by applying $H$ for a time
$T$ and then $-H$ for another $T$.  This retraces the path, makes the toggling-frame error symmetric in
time, and cancels the second-order term, so the infidelity drops to order $\epsilon^6$
(Proposition~\ref{prop:pm} in Appendix~\ref{app:magnus}).  We do not claim that this doubled sequence is
the fastest second-order robust control.

\paragraph*{Time-dependent noise.}  All results assume static or slow errors.  For noise that fluctuates
in time, the relevant quantity is the sensitivity of the control to each noise
frequency~\cite{Cywinski2008,Green2013,PazSilva2014}.  For the constant controls derived here, this
sensitivity vanishes at zero frequency and grows only quadratically at low frequency, so noise that is
slow compared with $T$ is still suppressed; the $\pm H$ sequence makes the growth quartic.  Since the
total sensitivity integrated over all frequencies is independent of the control~\cite{Cywinski2008},
noise much faster than the drive cannot be suppressed by any bounded control.

\paragraph*{Open questions.}  We determined the exact minimum time for classes that contain a full
level-shift class or a full subsystem class and have a common zero-diagonal basis.  For smaller
classes, such as a single diagonal error of unknown size on a qudit, the constant control gives an upper
bound and the optimal time is open.  The same holds for classes that need time-dependent control but are
smaller than complete ignorance, such as local errors on subsystems of different dimensions.  A complete
characterization of which error classes admit a common zero-diagonal basis would extend the criterion to
all classes; it connects robust control to open problems in simultaneous
hollowization~\cite{DammFassbender2020,Nicholus2026}.  Finally, our results concern the identity
operation; robust gates other than the identity remain to be treated.

\section{Conclusion}

Knowing where an error lies changes the form of an optimal robust control, not only its duration.  If
the error class has zero diagonal in some basis, an equally spaced energy ladder in that basis cancels
every error in the class to first order after one period.  For unknown level shifts and for unknown
local errors, no time-dependent control does better.  Only when the class is too large to admit such a
basis is time dependence needed, and complete ignorance then costs exactly $\sqrt{d^2+1}$ times more.

The practical recipe is short.  Find a basis in which the possible errors have no diagonal, and drive an
equally spaced energy ladder in that basis for one period.  For a qudit with drifting level shifts this
is a chiral drive of all transitions.  For two qubits with local field noise it is a constant
anisotropic $XX$--$ZZ$ coupling, which outperforms every single-qubit control at equal drive power.

\section*{Author contributions}
LL.B.G.\ conceived the project, defined its scope, and takes full responsibility for the paper. The work was carried out with the assistance of the language model Claude Opus~5.5, which helped develop the results, proposed candidate proofs that were then checked analytically and numerically, wrote and ran the
numerical code, and produced first drafts of the manuscript.

\section*{Acknowledgments}

This work has been supported by European Union’s HORIZON–MSCA-2023-DN-JD programme under the Horizon Europe (HORIZON) Marie Sklodowska-Curie Actions, grant agreement 101120296 (TENORS).

\section*{Data availability}

Code and data are available at the repository https://github.com/llorebaga/RobustSubspace.

\appendix

\section{First-order robustness and higher orders}\label{app:magnus}

\paragraph*{Magnus expansion.}  Let $\Xt(t)=U^\dagger(t)XU(t)$ be the toggling-frame error and
$V(t)=U^\dagger(t)U_\epsilon(t)$ the relative propagator.  It obeys $\dot V=-i\epsilon\Xt V$, and its
Magnus expansion is $V(T)=\exp[\epsilon M_1+\epsilon^2M_2+O(\epsilon^3)]$ with
\begin{align*}
 M_1&=-i\int_0^T\Xt\dd t,\\
 M_2&=-\frac12\int_0^T\!\dd t_1\!\int_0^{t_1}\!\dd t_2\,[\Xt(t_1),\Xt(t_2)] .
\end{align*}
Robustness means $M_1=0$; it does not imply $M_2=0$.

\paragraph*{Fidelity.}  Write $V(T)=e^{-iG}$ with $G$ Hermitian.  All Magnus terms are traceless (the
first because $X$ is, the others because they are commutators), so expanding the exponential gives
$F=|\tr V/d|^2=1-\|G\|_F^2/d+O(\|G\|^4)$.  Without robustness, $1-F=\epsilon^2\|M_1\|_F^2/d+O(\epsilon^3)$.
With robustness, $1-F=\epsilon^4\|M_2\|_F^2/d+O(\epsilon^5)$.

\paragraph*{Repeated periods.}  For a periodic control with $U(T)=e^{i\theta}I$ and a static error,
$V(nT)=V(T)^n=\exp[n\epsilon^2M_2+O(n\epsilon^3)]$.  The second-order error accumulates coherently, and
for a single error the fidelity at the peaks is quasi-periodic rather than monotonically decreasing.
Averaging over an ensemble of errors dephases these phases and produces the slow decay of the peaks in
Fig.~\ref{fig:revivals}.

\begin{proposition}[$\pm H$ sequence]\label{prop:pm}
Let $H$ be a constant Hamiltonian, robust to $S$, with $e^{-iHT}\propto I$.  Apply $H$ on $[0,T]$ and
$-H$ on $(T,2T]$.  For every static $X\in S$, the sequence returns exactly to the identity and
$M_1=M_2=0$ on $[0,2T]$, so $1-F=O(\epsilon^6)$.
\end{proposition}

\begin{proof}
For $t\in[T,2T]$, $U(t)=e^{iH(t-T)}e^{-iHT}=e^{-iH(2T-t)}=U(2T-t)$, so $U(2T)=I$ and
$\Xt(2T-t)=\Xt(t)$.  Hence $M_1$ on $[0,2T]$ is twice its first-period value, zero.  In $M_2$, substitute
$(t_1,t_2)=(2T-s_2,2T-s_1)$.  This maps the triangle $0\le t_2\le t_1\le2T$ onto itself with unit
Jacobian, while the integrand becomes $[\Xt(s_2),\Xt(s_1)]=-[\Xt(s_1),\Xt(s_2)]$.  Hence $M_2=-M_2=0$.
\end{proof}

The proof assumes a static error, an exact sign reversal of the whole control, and a robust first
period.  It is an instance of the general property of time-symmetric cycles in average Hamiltonian
theory~\cite{HaeberlenWaugh1968,Mansfield1971}.

\section{Proof of Theorem~\ref{thm:lower}}\label{app:lower}

\paragraph*{Regularity.}  For bounded measurable $\bm u$, the propagator is absolutely continuous with
bounded derivative, so every signal $\phi_{jk}$ lies in the Sobolev space $H^1([0,T])$ and obeys
$\phi_{jk}(T)=e^{i\theta}\phi_{jk}(0)$.  We take $H$ traceless; a trace part only adds a global phase and
increases $\|H\|_F$.

\paragraph*{Time modes.}  Let $B:\CC^{d^2}\to L^2([0,T])$ map the basis vector $(jk)$ to $\phi_{jk}$.  Then
$\Gamma=B^\dagger B$, while $\rho=BB^\dagger=\sum_{jk}|\phi_{jk}\rangle\langle\phi_{jk}|$ is an operator
on functions of time with the same nonzero eigenvalues.  Write
$\rho=\sum_{i=1}^r\gamma_i|e_i\rangle\langle e_i|$ with orthonormal modes $e_i$ and
$g\ge\gamma_1\ge\gamma_2\ge\dots\ge\gamma_r>0$, $\sum_i\gamma_i=d^2$.  Each $e_i$ is a combination of
signals, so it also lies in $H^1$ and obeys the twisted boundary condition.  Expanding
$\phi_{jk}=\sum_i\langle e_i,\phi_{jk}\rangle e_i$ and using
$\sum_{jk}\langle e_i,\phi_{jk}\rangle\overline{\langle e_{i'},\phi_{jk}\rangle}=\gamma_i\delta_{ii'}$ gives
\begin{equation}
 \sum_{jk}\|\dot\phi_{jk}\|^2=\sum_i\gamma_i\|\dot e_i\|^2 .
 \label{eq:kinrho}
\end{equation}

\paragraph*{Harmonics.}  The functions $f_s(t)=T^{-1/2}e^{i(2\pi s+\theta)t/T}$, $s\in\ZZ$, form an
orthonormal basis of $L^2([0,T])$.  For $f$ in $H^1$ with $f(T)=e^{i\theta}f(0)$, integration by parts
(the boundary terms cancel) and Parseval's identity give $\|\dot f\|^2=\sum_s\omega_s^2|\langle
f_s,f\rangle|^2$ with $\omega_s=(2\pi s+\theta)/T$.  Let $\lambda_1\le\lambda_2\le\dots$ be the values
$\omega_s^2$ in increasing order.

\paragraph*{Ky Fan's principle.}  For orthonormal $e_1,\dots,e_n$ as above,
$\sum_{i\le n}\|\dot e_i\|^2\ge\lambda_1+\dots+\lambda_n$.  Indeed,
$\sum_{i\le n}\|\dot e_i\|^2=\sum_s\omega_s^2w_s$ with $w_s=\sum_{i\le n}|\langle f_s,e_i\rangle|^2$.  Each
$w_s$ lies in $[0,1]$ (it is the squared length of the projection of $f_s$ onto the span of the $e_i$),
and $\sum_sw_s=n$, so the sum is smallest when weight one sits on the $n$ smallest $\omega_s^2$.  This
bounds partial sums only; it does not give $\|\dot e_i\|^2\ge\lambda_i$ for each $i$ (take $e_1=f_1$,
$e_2=f_0$, $\theta=0$).

\paragraph*{Abel summation.}  With $a_i=\|\dot e_i\|^2$ and $\gamma_{r+1}=0$,
\begin{align*}
 \sum_i\gamma_ia_i&=\sum_{n=1}^r(\gamma_n-\gamma_{n+1})\sum_{i\le n}a_i\\
 &\ge\sum_{n=1}^r(\gamma_n-\gamma_{n+1})\sum_{i\le n}\lambda_i=\sum_i\gamma_i\lambda_i ,
\end{align*}
where the inequality uses Ky Fan for every $n$ and needs the weights to be decreasing.

\paragraph*{Filling.}  For $0\le\gamma_i\le g$ with $\sum_i\gamma_i=gL$,
\begin{equation*}
 \sum_i\gamma_i\lambda_i-g\sum_{i\le L}\lambda_i\ge
 \lambda_L\Big[\sum_{i\le L}(\gamma_i-g)+\sum_{i>L}\gamma_i\Big]=0,
\end{equation*}
since $\gamma_i-g\le0$ and $\lambda_i\le\lambda_L$ for $i\le L$, while $\lambda_i\ge\lambda_L$ for $i>L$.
So the weight fills the $L$ lowest harmonics.

\paragraph*{Conclusion.}  The $L$ lowest values are $\lambda_i=(2\pi/T)^2y_i^2$ with
$y_i\in\ZZ+\theta/2\pi$ distinct, and Lemma~\ref{lem:spacing} gives
$\sum_{i\le L}\lambda_i\ge(2\pi/T)^2L(L^2-1)/12$.  Combining with Eqs.~\eqref{eq:energy}
and~\eqref{eq:kinrho},
$\frac d{2T}\int_0^T|\bm u|^2\dd t\ge g(2\pi/T)^2L(L^2-1)/12$, and $gL=d^2$ gives
$T\int_0^T|\bm u|^2\dd t\ge(2\pi)^2d(L^2-1)/6$.  With $\int_0^T|\bm u|^2\le\Om^2T$ this is
Eq.~\eqref{eq:lawL}; equality requires $|\bm u|=\Om$ throughout and, by Lemma~\ref{lem:spacing},
$\theta=0$ for odd $L$ and $\theta=\pi$ for even $L$.

\section{Proof of Proposition~\ref{prop:g}}\label{app:g}

\paragraph*{(i) Level shifts.}  A fixed change of basis $U\mapsto W^\dagger UW$ acts on $\Gamma$ by a
unitary similarity and changes neither the kinetic energy nor the endpoint, so we may take the basis to
be $\{|j\rangle\}$.  Inserting $X=|j\rangle\langle j|-|l\rangle\langle l|$ into Eq.~\eqref{eq:relations}
shows that the $d\times d$ blocks $(G_j)_{mn}=\Gamma_{(jm),(jn)}$ are all equal.  Unitarity says they sum
to $dI$, so each equals $I$.  Split $B$ into column blocks $B_j:\CC^d\to L^2$, $B_je_m=\phi_{jm}$; then
$B_j^\dagger B_j=I$, so each $B_j$ preserves lengths.  For $p=(p_1,\dots,p_d)$ with $p_j\in\CC^d$,
\begin{equation*}
 p^\dagger\Gamma p=\Big\|\sum_jB_jp_j\Big\|^2\le\Big(\sum_j\|p_j\|\Big)^2\le d\|p\|^2,
\end{equation*}
so $\Gamma\le dI$.

\paragraph*{(ii) One subsystem.}  Split the row index of $U$, on which the error acts, as $j=(a,b)$.  For
$X=X_A\otimes I$, Eq.~\eqref{eq:relations} reads
$\sum_{a,a'}(X_A)_{aa'}(\tr_b\Gamma)_{(am),(a'n)}=0$, where
$(\tr_b\Gamma)_{(am),(a'n)}=\sum_b\Gamma_{((ab)m),((a'b)n)}$ traces out the subsystem-$B$ part of the row
index.  Since this holds for all traceless $X_A$, the $d_A\times d_A$ matrix
$[(\tr_b\Gamma)_{(am),(a'n)}]_{aa'}$ is a multiple $c_{mn}$ of the identity, and unitarity gives
$c_{mn}=d_B\delta_{mn}$.  Thus $\tr_b\Gamma=d_B I$.  The lemma below, with the factor $b$ of dimension
$d_B$, gives $\Gamma\le d_B(\tr_b\Gamma)\otimes I=d_B^2I$.

\paragraph*{(iii) Complete ignorance.}  For fixed $m,n$, Eq.~\eqref{eq:relations} for all traceless $X$
forces $\Gamma_{(jm),(ln)}=c_{mn}\delta_{jl}$, and unitarity gives $c_{mn}=\delta_{mn}$, so $\Gamma=I$.

\begin{lemma}\label{lem:schmidt}
For a positive operator $M$ on $\mathcal H_1\otimes\mathcal H_2$,
$M\le\min(\dim\mathcal H_1,\dim\mathcal H_2)\,(\tr_2M)\otimes I_2$.
\end{lemma}

\begin{proof}
This is a form of the Schmidt-number criterion~\cite{TerhalHorodecki2000}.  Both sides add up over a
decomposition of $M$ into rank-one terms, so it suffices to take $M=|\psi\rangle\langle\psi|$ with Schmidt
decomposition $\psi=\sum_{i=1}^r\sqrt{\mu_i}\,a_i\otimes b_i$, $r\le\min(\dim\mathcal H_1,\dim\mathcal H_2)$.
For any $\chi$, Cauchy--Schwarz gives
$|\langle\chi,\psi\rangle|^2\le r\sum_i\mu_i|\langle\chi,a_i\otimes b_i\rangle|^2
\le r\sum_i\mu_i\langle\chi,(|a_i\rangle\langle a_i|\otimes I)\chi\rangle
=r\langle\chi,(\tr_2M\otimes I)\chi\rangle$.
\end{proof}

\section{Proof of Proposition~\ref{prop:universal}}\label{app:universal}

\paragraph*{Generator.}  Differentiating $U_*$ gives $\dot U_*=-iH_*U_*$ with the stated $H_*$, which is
smooth, Hermitian and traceless ($\tr B=\tr D=0$).  It is a single continuous trajectory, not a sequence
of two pulses.

\paragraph*{Averaging.}  Write $q_j=j-\frac{d-1}2$.  Since $\mathcal F^\dagger e^{-i\nu tB}=e^{-i\nu
tD}\mathcal F^\dagger$, the matrix $V=\mathcal F^\dagger U_*$ has elements
$V_{kj}(t)=(\mathcal F^\dagger)_{kj}e^{-i\nu t(q_k+dq_j)}$ with $q_k+dq_j=k+dj-\frac{d^2-1}2$.  As $(k,j)$
runs over all pairs, $k+dj$ runs through $0,\dots,d^2-1$ exactly once, so the $d^2$ elements oscillate at
distinct harmonics of $\nu=2\pi/T$, each with weight $|(\mathcal F^\dagger)_{kj}|^2=1/d$.  Hence
$\frac1T\int_0^T\overline{V_{kj}}V_{ln}\dd t=\frac1d\delta_{kl}\delta_{jn}$.  For any $X$, with
$Y=\mathcal F^\dagger X\mathcal F$,
$\frac1T\int_0^T(U_*^\dagger XU_*)_{jn}\dd t=\sum_{k,l}Y_{kl}\frac1d\delta_{kl}\delta_{jn}=\frac{\tr X}d\delta_{jn}$,
which vanishes for traceless $X$.

\paragraph*{Endpoint.}  $e^{-2\pi iq_j}=(-1)^{d-1}$ for every $j$, so $e^{-i\nu TB}=e^{-i\nu TD}=(-1)^{d-1}I$.
Also $e^{-2\pi idq_j}=e^{i\pi d(d-1)}=1$, so $e^{-id\nu TD}=I$, and $U_*(T)=(-1)^{d-1}I$.

\paragraph*{Amplitude.}  $\tr D^2=\tr B^2=d(d^2-1)/12$ and
$\tr(BD)=\sum_{jk}|\mathcal F_{jk}|^2q_jq_k=\frac1d(\sum_jq_j)^2=0$.  Since $B$ commutes with $e^{\pm i\nu tB}$,
$\tr H_*^2=\nu^2[\tr B^2+d^2\tr D^2+2d\tr(BD)]=\nu^2d(d^4-1)/12$ at all times, so
$|\bm u_*|^2=\nu^2d(d^4-1)/6$ and $|\bm u_*|T=2\pi\sqrt{d(d^4-1)/6}$.  The same holds with the two factors
of $U_*$ in reverse order.

\section{Two-qubit details}\label{app:twoqubit}

\paragraph*{Uniqueness of the coupling.}  The energies of $aXX+bYY+cZZ$ on $\Phi^+,\Phi^-,\Psi^+,\Psi^-$ are
$a-b+c$, $-a+b+c$, $a+b-c$ and $-a-b-c$.  Inverting this linear map for each of the $24$ assignments of
$\{\pm\frac32,\pm\frac12\}J$ gives $\{|a|,|b|,|c|\}=\{1,\frac12,0\}J$ in every case.  Signs can be changed
by local unitaries, and permutations by relabelling axes.

\paragraph*{Protected errors.}  The Bell states are joint eigenstates of $XX$, $YY$ and $ZZ$.  A Pauli
operator that anticommutes with one of these has zero expectation in every Bell state.  This includes all
six local Pauli operators and $XY,XZ,YX,YZ,ZX,ZY$, so $H_{\rm ex}$ cancels this $12$-dimensional class.
For $H'=(I\otimes W)H_{\rm ex}(I\otimes W)$, with $W$ a Hadamard, the eigenstates are
$(I\otimes W)|\beta_k\rangle$ and $\langle\beta_k|(I\otimes W)ZZ(I\otimes W)|\beta_k\rangle=\langle\beta_k|ZX|\beta_k\rangle=0$.

\paragraph*{Local-only controls.}  For $H=H_1\otimes I+I\otimes H_2$ with $H_i=\frac12\bm h_i\cdot\bm\sigma$,
the propagator factorizes, $U=U_1\otimes U_2$, and robustness to $X_A\otimes I$ requires
$\int_0^TU_1^\dagger X_AU_1\dd t=0$ for every single-qubit $X_A$.  Proposition~\ref{prop:g}(iii) and the
energy bound of Theorem~\ref{thm:lower} with $d=2$, $L=4$ give $T\int_0^T|\bm h_i|^2\dd t\ge20\pi^2$ for
each qubit.  Since $|\bm u|^2=2(|\bm h_1|^2+|\bm h_2|^2)$, we get $\Om^2T^2\ge80\pi^2$, i.e.,
$\Om T\ge2\pi\sqrt{20}$.  The bound is attained by driving each qubit with
$\bm h(t)=\frac{\Om}{2\sqrt5}\big(2,-\sin(4\pi t/T),\cos(4\pi t/T)\big)$, the $d=2$ path of
Proposition~\ref{prop:universal} (factors in reverse order, axes relabelled), with $|\bm h_i|=\Om/2$.  This
is the local reference in Fig.~\ref{fig:revivals}(b).

\section{Numerical methods}\label{app:numerics}

All simulations use $\Om=1$.  An error is written $\epsilon X$ with $X=\frac12\sum_ac_a\lambda_a$, where
$\{\lambda_a\}$ is an orthonormal basis of the class ($\tr\lambda_a\lambda_b=2\delta_{ab}$), so $\epsilon$ is
measured in the same units as $|\bm u|$.  For two qubits the in-class basis is
$\{XI,YI,ZI,IX,IY,IZ\}/\sqrt2$.  ``Outside the class'' means errors drawn from the full Gell-Mann basis of
$\mathfrak{su}(d)$.  Constant controls are propagated exactly by diagonalization; time-dependent controls by
products of exact exponentials at the midpoint of each step.

\paragraph*{Fig.~\ref{fig:revivals}(a,b).}  $200$ realizations per curve ($80$ for the local reference),
with independent normal $c_a$ of variance $\epsilon^2/n_S$ ($n_S$ basis elements), so that
$\langle|\bm c|^2\rangle=\epsilon^2$, and $\epsilon=0.1$; realizations are not normalized individually.  The
local reference uses steps of length $T/400$.

\paragraph*{Fig.~\ref{fig:revivals}(c).}  For seven values of $\epsilon$ between $10^{-2.5}$ and
$10^{-0.7}$, the infidelity after one period ($2T$ for $\pm H$) is maximized over $24$ random unit vectors
$\bm c$.  Slopes are least-squares fits of $\log(1-F)$ against $\log\epsilon$ for $\epsilon\le0.06$ and
$1-F>10^{-13}$; uncertainties are standard errors of the slope (Table~\ref{tab:slopes}).

\begin{table}[h]
\caption{Fitted exponents $p$ in $1-F\propto\epsilon^p$ for Fig.~\ref{fig:revivals}(c).}
\label{tab:slopes}
\begin{ruledtabular}
\begin{tabular}{lcc}
control, errors & $p$ & expected\\
\hline
chiral drive, in class & $3.99994\pm0.00002$ & $4$\\
chiral drive, outside & $1.993\pm0.003$ & $2$\\
$XX$--$ZZ$ coupling, in class & $3.996\pm0.002$ & $4$\\
$XX$--$ZZ$ coupling, outside & $2.002\pm0.001$ & $2$\\
chiral, $\pm H$, in class & $5.998\pm0.001$ & $6$\\
$XX$--$ZZ$, $\pm H$, in class & $5.990\pm0.004$ & $6$
\end{tabular}
\end{ruledtabular}
\end{table}

\paragraph*{Minimum-time search.}  Controls are piecewise constant with eight equal segments in all
$d^2-1$ directions of $\mathfrak{su}(d)$, with $\bm u_j=\bm v_j/\sqrt{1+|\bm v_j|^2}$ so that $|\bm u_j|<\Om$.
The constraints are robustness for a basis of the class (six-point Gauss--Legendre quadrature per segment)
and $U(T)\propto I$.  Each start first solves the constraints at a fixed duration $1.25$--$1.75$ times the
bound, from random $\bm v_j$, by nonlinear least squares; it then minimizes $T$ by sequential quadratic
programming (SLSQP, tolerance $10^{-13}$) with exact Jacobians from automatic differentiation.  Qutrit level
shifts: $11$ starts, all converged, $T/T_{\min}$ between $1.30$ and $1.69$.  Two-qubit local errors: $8$
starts, $7$ converged, $T/T_{\min}$ between $1.53$ and $1.68$.  Re-evaluating robustness with $60$-point
quadrature gives residuals below $2\times10^{-8}$.

\bibliography{refs}

\end{document}